\documentclass[11pt,onecolumn,singlespace,conference,compsoc]{IEEEtran}

\usepackage{comment}
\usepackage{enumerate}
\usepackage{color}
\usepackage{relsize}
\usepackage[dvipsnames]{xcolor}
\usepackage{subfigure}
\usepackage{amsmath}
\usepackage{epsfig}
\usepackage{url}
\usepackage{amssymb}
\usepackage{amsmath,amssymb,amsfonts}
\usepackage{amsthm}
\usepackage{mathtools}
\usepackage{tikz-cd}

\newtheorem{proposition}{Proposition}

\theoremstyle{definition}
\newtheorem{definition}{Definition}[section]

\theoremstyle{lemma}
\newtheorem{lemma}{Lemma}[section]

\begin{document}

\title{Topological Properties of Multi-Party Blockchain Transactions}

\author{Dongfang Zhao\\ 
Department of Computer Science and Engineering\\
The University of Nevada\\
Reno, NV 89557, United States\\
dzhao@unr.edu}

\maketitle

\begin{abstract}
The cross-blockchain transaction remains one of the most challenging problems in blockchains.
The root cause of the challenge lies in the nondeterministic nature of blockchains:
A $n$-party transaction across multiple blockchains might be partially rolled back due to the potential forks in any of the participating blockchains---eventually, only one fork will survive in the competition among miners.
While some effort has recently been made to developing hierarchically distributed commit protocols to make multi-party transactions progress,
there is no systematic method to reason about the transaction outcome.
This paper tackles this problem from a perspective of point-set topology.
We construct multiple topological spaces for the transactions and blockchain forks,
and show that these spaces are internally related through either homeomorphism or continuous functions.
Combined together, these tools allow us to reason about the cross-blockchain transactions through the growing-fork topology, an intuitive representation of blockchains.
\end{abstract}

\vspace{5mm}

\section{Introduction}  

Cross-blockchain transaction remains one of the most challenging problems in blockchains~\cite{dzhao_cidr20}.
The root cause of the challenge lies in the nondeterministic nature of blockchains:
A transaction across multiple blockchains might be partially rolled back due to the potential forks in any of the participating blockchains---eventually, only one fork will survive in the competition among miners.
While some effort has recently been made to developing hierarchically distributed commit protocols~\cite{dzhao_arxiv2002_cbt,vzakhary_arxiv19,mherlihy_podc18} to ensure the liveness of multi-party transactions,
they lack a theoretical foundation to guarantee the \textit{completeness} of the transactions, either finally committed or aborted.
That is, the state-of-the-art solely depends on concrete blockchain implementations (e.g., ~\cite{ethereum,bitcoin,algorand_sosp17,hyperledger_eurosys18,hotstuff_podc19})) without formal methods and abstractions.

This paper tackles the $n$-blockchain-transaction problem from the perspective of point-set topology.
We construct a topological space for the transactions among blockchains
and show that this space is homeomorphic to another topological space of static fork graphs induced by those transactions.
We then construct a time-varying counterpart of the static fork graphs, namely the growing-fork topological space,
and show that a continuous function exists from the growing-fork topology to the static forks.
Combined together, these tools allow us to reason about the cross-blockchain transactions through the growing-fork topology, an intuitive representation of blockchains.
For instance, the \textit{decidiability} of whether an arbitrary $n$-party cross-blockchain transaction is able to complete in finite time can be reduced to a pure topological problem where we can call upon a vast literature of results in combinatorial and algebraic topology.

The remainder of this paper is organized as follows.
We review elementary concepts of blockchains and topology in~\S\ref{sec:bg}.
The system models and assumptions of the proposed topological approach to cross-blockchain transactions are provided in~\S\ref{sec:model}.
We detail the construction of topological spaces for transactions and blockchain forks in~\S\ref{sec:space}.
The relationships among the topological spaces are then illustrated through multiple continuous functions in~\S\ref{sec:map}.
We finally conclude this paper in~\S\ref{sec:final}.

\section{Background and Related Work}
\label{sec:bg}

\subsection{Blockchains}

\textbf{Data Structures.}
A blockchain is a linkedlist replicated on multiple machines.
Machines are usually also called nodes, miners, or participants.
The element of the linkedlist is a block of multiple transactions,
each of which involves two encrypted addresses and a specific number (e.g., the Bitcoin amount to transfer).
The block's overall data, plus a puzzle number, called a \emph{nonce}, are hashed into a value that is passed to the next block,
which will do the same and pass along the hashed value.
Therefore, a linked list of ``locked'' blocks is formed through the compounded hash values.

\textbf{Mining.}
The procedure to solve the puzzle, i.e., find the nonce, is called mining;
For instance, Bitcoin requires that the hashed value of the nonce and current transactions is below some threshold.
Because of this inequality requirement, the qualified nonce value is not unique;
It is possible that two or more nodes solve the puzzle at the same time,
and this is allowed in blockchains.
The consequence is then there might exist more than one linkedlists,
called forks.
Depending on implementations, the number of forks is expected to reduce back to one.
For instance, Bitcoin achieves this by eliminating those forks that are shorter than others.
Initially, a blockchain has only one linkedlist, initiated by the very first block, called the genesis block.

\textbf{Pools.}
When a new block is appended to the linkedlist, a system reward plus some transaction fee will be transferred to the (encrypted address) of the miner.
To improve the chance to win the race of mining, a miner can choose to join a pool of other miners to share the reward with the members in the pool.
Since the success of making a profit is all about competition, it is not uncommon for a pool to launch an attack again other pools.

\subsection{Modeling Blockchains}

\textbf{Game theory.}
There are some work leveraging game theory to study the behaviors of blockchains, especially the pooling formation in cryptocurrency.
In~\cite{ieyal_sp15}, Eyal studied how one pool decides to attack another pool by sending spy nodes through a nonoperative game.
In a similar spirit, Tsabary and Eyal~\cite{itsabary_ccs18} showed that a game-theoretical model could be built for the scenarios when transaction fees play a more important role in miner's behaviors.

\textbf{Group theory.}
Speaking of blockchain pools, another branch of study was to take a group-theoretical approach~\cite{dzhao_arxiv2002_gt}.
Essentially, the ``movement'' of miners is modeled as a permutation of the set associated with the blockchain.
Such movement was then shown to be group operation:
closed, associative, and inversable.
Therefore, a lot of group properties are immediately available to the pool group of blockchains.

\subsection{Cross-blockchain Operations}

The first study on operations among an arbitrary number of blockchains was published in~\cite{mherlihy_podc18}.
Herlihy showed that a simple timeout scheme could enable the atomicity in a multi-party operation.
However, the operation is not necessarily a transaction:
partial changes are still possibly committed.
Admittedly, it is arguable that not all applications require strong transactional properties.
The latest findings in this direction can be found in the followup work~\cite{mherlihy_vldb19}.

Another branch of study over cross-blockchain operations indeed focuses on transactions.
A recent work called AC3~\cite{vzakhary_arxiv19} employs an extra component (known as \textit{witness blockchain}) to govern the cross-chain operations.
Although the witness blockchain is comprised of the nodes from existing blockchains, still, these \textit{virtual} nodes on the witness blockchain become the critical components of the entire ecosystem.
To overcome this issue, a completely distributed commit protocol was proposed in~\cite{dzhao_arxiv2002_cbt}.

\subsection{Topology}

We conclude this section by reviewing some basic topology concepts and theorems.
A topology of a set $S$ is a collection of subsets of $S$, denoted $\mathcal{T}$.
One example topology of $S$ is then the power set of $S$, $\mathcal{P}(S)$, 
which consists of all the possible $2^{|S|}$ subsets of $S$.
This is also called the \textit{discrete topology} of $S$.
The tuple $(S, \mathcal{T})$ is called the \textit{topological space} of $S$.
If the context is clear, we often refer to $S$ to indicate space $(S, \mathcal{T})$.
Each of the subsets $U$ from $\mathcal{T}$ is called an \textit{open set}, and the complement set $S \setminus U$ is a \textit{closed set} by definition.
A function $g$ from space $X$ to $Y$ is called \textit{continuous} if $\forall v$ is an open set in $Y$, then $g^{-1}(v)$ is an open set in $X$.
The composition of two continuous functions is also continuous.
If both $g$ and $g^{-1}$ are continuous, we call $g$ a \textit{homeomorphism}.
Because a homeomorphism is defined purely on open and closed sets,
two topological spaces are considered equivalent if such homeomorphism exists.
Usually, we expect to migrate a complex problem in one topological space to another such that the problem can be solved more efficiently or more intuitively.
The aforementioned concepts and techniques are also referred to as \textit{point-set topology}.

Point-set topology is the main technique used in this paper;
leveraging point-set topology in distributed computing contexts dated back in 1980's~\cite{balpern_ipl85} and recently revived in~\cite{tnowak_podc19}.
We will show how to reason the cross-blockchain transaction problem with a series of continuous functions constructed from the more intuitive topological properties on the blockchain forks.
To the best of our knowledge, this paper is the first work taking a point-set topological approach to study multi-party transactions among blockchains.

In addition to point-set topology, there are \textit{algebraic-topological} methods to study those problems that are better modeled in a geometric sense.
These methods, usually categorized into \textit{homotopy} groups and \textit{homology} groups,
study the ``smaller'' pieces of the targeting objects and try to ``map'' the geometrical objects into algebraic objects, such as \textit{groups}.
The smaller pieces are \textit{loop curves} and \textit{triangle patches} (\textit{simplices}, more formally) for homotopy groups and homology groups, respectively.
Then, the equivalence between two topological spaces can be investigated by checking the algebraic groups, usually through \textit{homomorphic groups}.
A good review of algebraic-topological techniques applied to distributed computing can be found in~\cite{mherlihy_book13};
some of the hardest problems were shown to be elegantly solvable through algebraic topology~\cite{mherlihy_dc13,mherlihy_podc10,rguer_tcs09,mherlihy_jacm99}.
Remarkably, a unique subbranch of topology, namely \textit{combinatorial topology},
specifically studied the topological properties of distributed computing models~\cite{mherlihy_stoc93,mherlihy_stoc94}.
This paper does not touch any of these algebraic- or combinatorial-topological methods,
although we might explore them in the near future.

\section{Models and Assumptions}
\label{sec:model}

We assume each blockchain has a nontrivial number of nodes.
That is, each blockchain is a \textit{cluster} of nodes.
We will use blockchain and cluster interchangeably.
The set of blockchains is denoted $C$,
where each blockchain is $C_i$, $0 \leq i < n$, $i \in \mathbb{Z}$.

We assume each cluster can spawn an arbitrary number of forks.
Although two forks are most commonly seen in cryptocurrency,
it is not uncommon the have more forks if the underlying consensus is customized for domain-specific usage, e.g., scientific data provenance~\cite{aalmamun_bigdata18}.
A \textit{blockchain fork} is defined as follows in this paper.

\begin{definition} [Blockchain Fork]
Let $F$ denote the set of all forks,
each of which, denoted $F_i$, resides on a cluster $C_i$.
Initially, all $F$'s have a single fork, 
initiated by the genesis block.
When $f$ ($f \in \mathbb{Z}, f > 1$) nodes succeed in a specific round,
the blockchain will spawn $f-1$ new forks ($0 \leq i \leq f-1$).
We use $F_{-i} = F \setminus \{F_{i}\}$ to denote the complement set of fork $F_{i}$ in $F$.
Each fork will have one of the following three states.
$F_i$ is called \textit{eliminated} if any of other forks $F_j \in F_{-i}$ suppress $F_i$.
$F_i$ is \textit{confirmed} if all other forks $F_{-i}$ are eliminated.
Any other forks are called \textit{undecided}.
\end{definition}

With this definition, each fork can be categorized into one of the three possible states. 
Therefore, we can construct a \textit{fork graph} among the elements of $F$ with three types of vertices/nodes.
Note that this graph is static; we speak of nothing about the timestamp or step number in this definition.
Now we are ready to define the transactions among these forks.

\begin{definition} [Transaction Proxy]
A transaction proxy on a blockchain, or, equivalently, a fork graph, is a node of that fork graph where the transaction is carried out.
A proxy is called \textit{live} if its fork is pending or confirmed.
\end{definition}

Therefore, a transaction is a collection of states on the proxies from the fork graphs.
The transaction is time-oblivious since users are only interested in the final result of the transaction: \textit{commit}, \textit{pending}, or \textit{abort}.

Nevertheless, a blockchain is indeed a dynamic data structure that changes over time.
To this end, we introduce the extended concept of the fork graph, the so-called \textit{growing-fork graph}.

\begin{definition} [Growing-fork Graph]
A growing-fork graph, $F^\omega_i$, is a sequence of fork graphs associated to the cluster $C_i$.
Let $F_i^t$ denote the fork graph of cluster $C_i$ at time $t$,
then 
\[ \displaystyle
F^\omega_i = \left( F_i^0, F_i^1, \dots \right).
\]
\end{definition}

With the above terms defined, we are ready to construct the topological spaces associated to the transaction proxies, the static fork, and the real-time blockchain forks.

\section{Topological Spaces}
\label{sec:space}

\subsection{Topological Space of Cross-Blockchain Transactions}

\begin{definition}
Let $C$ denote the set of clusters each of which represents a blockchain, $C = \{C_0, \dots, C_{n-1}\}$.
Define the \textit{ternary distance} between any pair of elements in $C$, $C_i$ and $C_j$, as a map $d_t: C \times C \rightarrow \{1,\frac{1}{2},2\}$ as follows:
\[
  d_t(C_i, C_j) =
    \begin{cases}
      2, & \text{both $C_i$ and $C_j$ commit}; \\
      \frac{1}{2}, & \text{both $C_i$ and $C_j$ abort}; \\
      1, & \text{otherwise}.
    \end{cases}  
\]
\end{definition}

We illustrate the three scenarios in Figure~\ref{fig:topo_t}.
The two proxies in the green transaction can both commit since there is no fork at the moment.
The two proxies in the orange transaction can safely abort because both forks from $C_0$ and $C_1$ are to be eliminated.
We cannot decide the result of the yellow transaction because there are forks involved.
A transaction might involve more than two clusters, say $n > 2$, $n \in \mathbb{Z}$.
The states of a $n$-party transaction, i.e., a $n$-cluster transaction, is similarly defined as two parties.

\begin{figure}[!t]
    \centering
    \includegraphics[width=80mm]{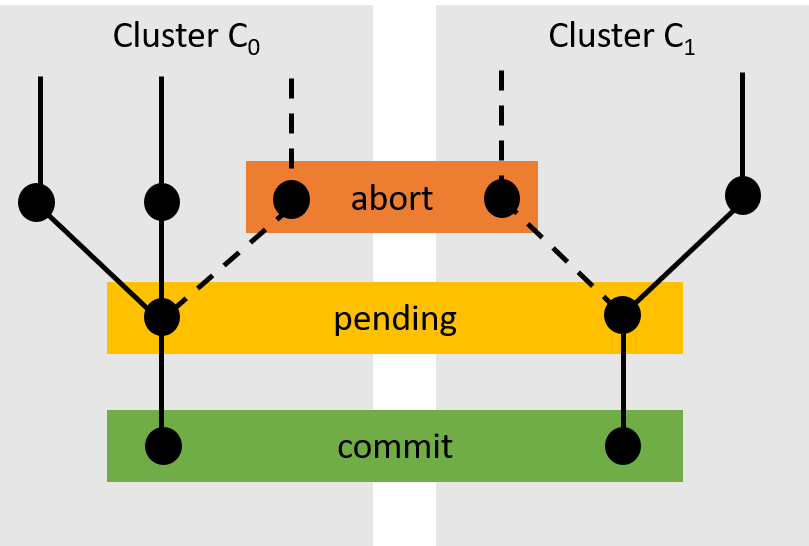}
    \caption{Example of three possible results/distances for a two-cluster transaction.
    }
    \label{fig:topo_t}
\end{figure}

\begin{lemma}
$d_t$ is a metric on $C$.
\end{lemma}
\begin{proof}
Obviously, $d_t(C_i, C_j) = d_t(C_j, C_i) \geq 0$ by definition.
It remains to show $d_t(C_i, C_j) \leq d_t(C_i, C_k) + d_t(C_k, C_j)$,
for arbitrary $k \in \mathbb{Z}$, $0 \leq k < n$.

If $d_t(C_i, C_j) = \frac{1}{2}$, it means both $C_i$ and $C_j$ abort.
If $C_k$ commits, then
\[
d_t(C_i, C_k) + d_t(C_k, C_j) = 1 + 1 > \frac{1}{2} = d_t(C_i, C_j).
\]
If $C_k$ aborts, then
\[
d_t(C_i, C_k) + d_t(C_k, C_j) = \frac{1}{2} + \frac{1}{2} > \frac{1}{2} = d_t(C_i, C_j).
\]

If $d_t(C_i, C_j) = 1$, then we know one blockchain commits the transaction and the other aborts. 
Without loss of generality, we assume $C_i$ commits and $C_j$ aborts in the following.
If $C_k$ commits, then 
\[
d_t(C_i, C_k) + d_t(C_k, C_j) = 2 + 1 > 1 = d_t(C_i, C_j).
\]
If $C_k$ aborts, then
\[
d_t(C_i, C_k) + d_t(C_k, C_j) = 1 + \frac{1}{2} > 1 = d_t(C_i, C_j). 
\]
Therefore, $d_t(C_i, C_j) \leq d_t(C_i, C_k) + d_t(C_k, C_j)$.

If $d_t(C_i, C_j) = 2$, meaning that both $C_i$ and $C_j$ commit, 
it is also easy to check the triangle inequality.
If $C_k$ commits, then
\[
d_t(C_i, C_k) + d_t(C_k, C_j) = 2 + 2 > 2 = d_t(C_i, C_j).
\]
If $C_k$ aborts, then
\[
d_t(C_i, C_k) + d_t(C_k, C_j) = 1 + 1 = 2 = d_t(C_i, C_j).
\]
Therefore, again, $d_t(C_i, C_j) \leq d_t(C_i, C_k) + d_t(C_k, C_j)$.
\end{proof}

\begin{proposition} \label{thm:txn_space}
Topological space $\mathcal{T}_{d_t}^C$ on $C$ is induced by $d_t$.
\end{proposition}

\begin{proof}
Let $\epsilon = \frac{1}{4}$, then the open ball $B_{d_t}(u, \epsilon)$ around the origin $u$, an open set in the topology.
Because all the possible distances between any clusters are at least $\frac{1}{2}$, 
this $\epsilon$-ball $B_{d_t}(u, \epsilon)$ splits the cluster $C$ into a collection of single-element sets,
each of which is a set of a single cluster $C_i$, $0 \leq i < n$.
That is, 
\[
\mathcal{B} = \{ \{C_0\}, \dots, \{C_{n-1}\}\}.
\]
Evidently, $\mathcal{B}$ is basis of $\mathcal{T}_{d_t}^C$:
every element in $C$ exactly belongs to an element in $\mathbb{B}$,
and an open set in the topology is simply an arbitrary union of $b \in \mathcal{B}$.
\end{proof}

Having constructed the topology of transactions among blockchains, 
we will start building a topology from the perspective of blockchain forks.

\subsection{Topological Space of Fork Graphs}

\begin{definition}[Fork Distance]
When two transaction proxies are both live on two fork graphs $F_i$ and $F_j$, respectively,
the fork distance $d_f$ is defined as
\[
d_f (F_i, F_j) = \frac{1}{|F_i|} + \frac{1}{|F_j|}.
\]
If any of the two proxies is not live,
we define $d_f = \delta_f = \displaystyle \frac{1}{\mathtt{sup}\{|F_i|: F_i \in F\}}$.
\end{definition}

As an example, we calculate some fork distances of transactions on Figure~\ref{fig:topo_t} as follows ($F_0$ and $F_1$ representing the fork graphs of $C_0$ and $C_1$, respectively):
\begin{itemize} \setlength\itemsep{1em}
    \item 
    $\displaystyle
    d_f(F^{commit}_0, F^{commit}_1) = \frac{1}{|F^{commit}_0|} + \frac{1}{|F^{commit}_1|} = \frac{1}{2} + \frac{1}{2} = 1; 
    $
    \item 
    $\displaystyle
    d_f(F^{abort}_0, F^{abort}_1) = \frac{1}{\mathtt{sup}\{|F^{abort}_0|, |F^{abort}_1|\}} = \frac{1}{\mathtt{sup}\{3, 2\}} = \frac{1}{3}.
    $
\end{itemize}

Now, we will show that $d_f$ is a metric on the set of fork graphs.

\begin{lemma}
Fork distance $d_f$ is a metric on $F$.
\end{lemma}
\begin{proof}
If any proxy is not live, say that on $F_i$, then
(i) $d_f (F_i, F_j) \triangleq \delta_f \geq 0$,
(ii) $d_f (F_i, F_j) = d_f (F_j, F_i) = \delta_f$, and
(iii) $d_f (F_i, F_j) = \delta_f \leq \delta_f + d_f (F_k, F_j) = d_f (F_i, F_k) + d_f (F_k, F_j)$.

If both proxies are live, then:
\begin{itemize}
    \item $d_f > 0 \geq 0$ by definition;
    \item $\displaystyle d_f(F_i, F_j) = \frac{1}{|F_i|} + \frac{1}{|F_j|} = \frac{1}{|F_j|} + \frac{1}{|F_i|} 
    = d_f(F_j, F_i)$;
    \item $\displaystyle d_f(F_i, F_j) 
    < \frac{1}{|F_i|} + \frac{1}{|F_j|} + \frac{2}{|F_k|}
    = \left(\frac{1}{|F_i|} + \frac{1}{|F_k|}\right) + \left(\frac{1}{|F_k|} + \frac{1}{|F_j|}\right)
    = d_f(F_i, F_k) + d_f(F_k, F_j)$.
\end{itemize}
\end{proof}

Now we are ready to construct the topology of fork graphs.

\begin{proposition} [Fork Space]
\label{thm:fork_space}
The topological space $\mathcal{T}_{d_f}^F$ over the fork set $F$ is induced by $d_f$.
\end{proposition} 

\begin{proof}

Let $\displaystyle \epsilon = \frac{1}{1 + \mathtt{sup}\{|F_i|: F_i \in F\}}$. Then let $B_{d_f}(F_i, \epsilon)$ be an open ball around the center of $F_i$ with a radius of $\epsilon$, namely an \textit{$\epsilon$-ball}.
We will show that these $\epsilon$-balls form a \textit{discrete topology}:
every open set induced by an open $\epsilon$-ball is a singleton subset of exactly one fork graph in $F$.
That is, we need to show that $\epsilon < d_f$ for any pair of fork graphs.
Note that by definition, $\epsilon < \delta_f$.
Therefore, it suffices to show that $\delta_f < d_f$.
Indeed, that is how we construct $\delta_f$, because
\[
\delta_f
< 2 \delta_f
= \displaystyle \frac{2}{\mathtt{sup}\{|F_i|: F_i \in F\}}
= \frac{1}{\mathtt{sup}\{|F_i|: F_i \in F\}} + \frac{1}{\mathtt{sup}\{|F_j|: F_j \in F\}}
\leq \frac{1}{|F_i|} + \frac{1}{|F_j|}
= d_f(F_i, F_j).
\]

Therefore, $\epsilon$ is ``small'' enough to split the space into a series of elements whose distance exceeds the boundaries of those $\epsilon$-balls.
That is, we now have a collection $\mathcal{B}$ of singleton open sets, each of which comprises exactly one fork graph:
\[
\mathcal{B} = \{ \{F_0\}, \dots, \{F_{n-1}\}\}.
\] 
$\mathcal{B}$ is a basis because any element in $F$ belongs to a subset of $\mathcal{B}$ and the intersection between any two subsets is empty.
\end{proof}

Now we are ready to study the time-varied topology in terms of blockchain forks, the so-called \textit{growing-fork} graphs.
Essentially, we extended the fork space by the time dimension.

\subsection{Topological Space of Growing-fork Graphs}

We extend the fork graph with a timestamp to model the real-time fork topology.
A fork graph at time $t$ is denoted $F^t$;
the set of infinite fork graphs is denoted $F^\omega$, by the naming convention of point-set topology.
That is, $F^\omega = (F^0, F^1, \dots, )$.
$F^\omega_i$ then represents the ever-growing fork topology of cluster $C_i$.

As before, we first define the metric over the growing-forks.

\begin{definition} [Growing fork distance $d_g$]
\[\displaystyle
d_g (F^\omega_i, F^\omega_j) 
= \frac{1}{\mathtt{inf} \{ |F^m_i|, |F^m_j| \} },
\]
where $m = \mathtt{inf}\{t: |F^t_i| \not= |F^t_j| \}$, $t \in \mathbb{Z}$, $t \ge 0$.
Essentially, $d_g$ tracks the smallest possible number of forks between two growing-fork graphs once any of the two starts forking.
Note that all blockchains initially have a single fork initiated by the genesis block.
\end{definition}

To make matters more concrete, we illustrate the metric in Figure~\ref{fig:topo_g},
where there are three blockchains ($F_0, F_1, F_2$) in the first three steps ($t = 0, 1, 2$).
Here are some example calculations:
\begin{itemize} \setlength\itemsep{1em}
    \item $\displaystyle
        d_g(F_0^\omega, F_1^\omega) 
        = \frac{1}{\mathtt{inf}\{|F^2_0|, |F^2_1|\}}
        = \frac{1}{\mathtt{inf}\{3, 2\}}
        = \frac{1}{2}$;
        
    \item $\displaystyle
        d_g(F_1^\omega, F_2^\omega) 
        = \frac{1}{\mathtt{inf}\{|F^2_1|, |F^2_2|\}}
        = \frac{1}{\mathtt{inf}\{2, 1\}}
        = \frac{1}{1}$
        = 1;        
        
    \item $\displaystyle
        d_g(F_0^\omega, F_2^\omega) 
        = \frac{1}{\mathtt{inf}\{|F^2_0|, |F^2_2|\}}
        = \frac{1}{\mathtt{inf}\{3, 1\}}
        = \frac{1}{1}$
        = 1;   
\end{itemize}

\begin{figure}[!t]
    \centering
    \includegraphics[width=150mm]{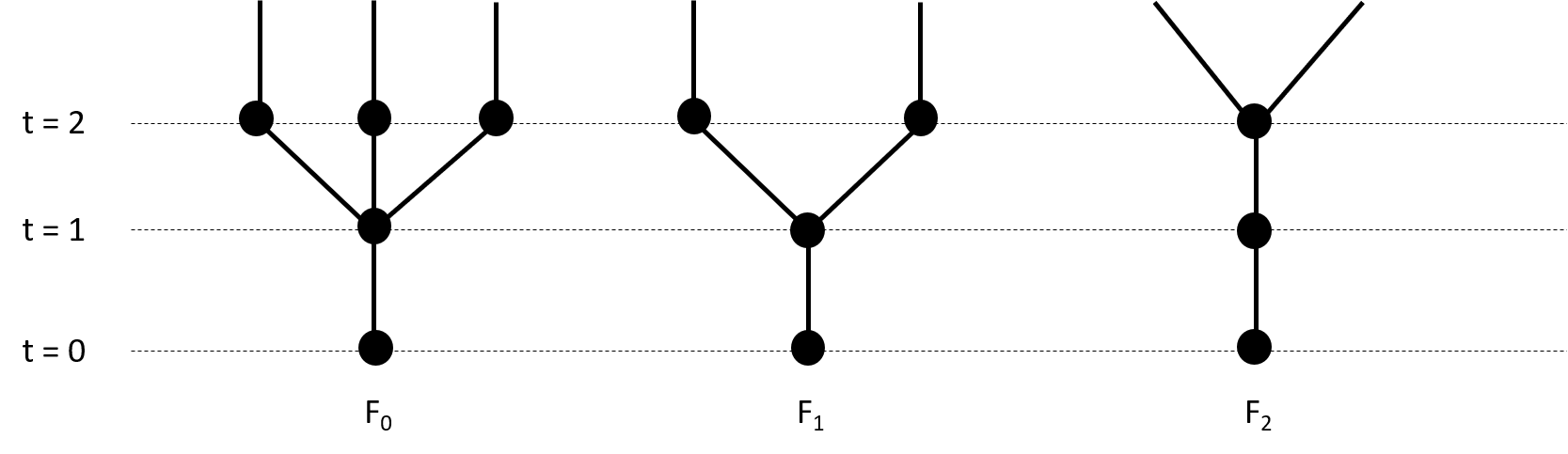}
    \caption{Example of three blockchains' growing-fork topology.
    }
    \label{fig:topo_g}
\end{figure}

As before, we will prove that $d_g$ is indeed a metric:
\[
\forall x, y, z \in F^\omega, \mathtt{ then }\hspace{1mm} d_g(x) \ge 0, d_g(x,y) = d_g(y,x), d_g(x,y) \leq d_g(x,z) + d_g(z,y).
\]

\begin{lemma}
$d_g$ is a metric on $F^\omega$.
\end{lemma}
\begin{proof}
Evidently, 
\[
d_g(F_i^\omega, F_j^\omega) = d_g(F_j^\omega, F_i^\omega) = \frac{1}{\mathtt{inf} \{ |F^m_i|, |F^m_j| \} } \ge 0,
\]
by definition.
It remains to show $d_g(F_i^\omega, F_j^\omega) \leq d_g(F_i^\omega, F_k^\omega) + d_g(F_k^\omega, F_j^\omega)$.
Define $m_i$ to be the smallest index such that $|F_i^{m_i}| > 1$.

If $m_i = m_j$, it means $m = m_i = m_j$.
\begin{itemize}

\item If $m_k < m$, then we know that $F_k^\omega$ starts to fork earlier than $F_i^\omega$ and $F_j^\omega$.
That makes 
\[\displaystyle
d_g(F_i^\omega, F_k^\omega) = d_g(F_k^\omega, F_j^\omega) = \frac{1}{|F_i^{m_i}|} = 1.
\]
Note that since both $F_i^\omega$ and $F_j^\omega$ start forking at the same time, we have
\[\displaystyle
d_g(F_i^\omega, F_j^\omega) \leq \frac{1}{2} < 2 = d_g(F_i^\omega, F_k^\omega) + d_g(F_k^\omega, F_j^\omega). 
\]

\item If $m_k > m$, this means $F_k^\omega$ starts to fork after $F_i^\omega$ and $F_j^\omega$. 
Then the two distances $d_g(F_i^\omega, F_k^\omega)$ and $d_g(F_k^\omega, F_j^\omega)$ depend on $F_k^{m_k}$.
We can similarly calculate that 
\[\displaystyle
d_g(F_i^\omega, F_k^\omega) = d_g(F_k^\omega, F_j^\omega) = \frac{1}{|F_i^{m_i}|} = 1,
\]
and draw the same conclusion as above.

\item If $m_k = m$, then all three growing-forks start forking at the same time. 
Then we have
\[\displaystyle
d_g(F_i^m, F_k^m) = \frac{1}{\mathtt{inf} \{ |F^m_i|, |F^m_k| \} } \geq \frac{1}{|F^m_i|}.
\]
Similarly, we also have
\[\displaystyle
d_g(F_k^m, F_j^m) = \frac{1}{\mathtt{inf} \{ |F^m_k|, |F^m_j| \} } \geq \frac{1}{|F^m_j|}.
\]
Therefore, we have 
\[\displaystyle
d_g(F_i^m, F_k^m) + d_g(F_k^m, F_j^m) \geq \frac{1}{|F^m_i|} + \frac{1}{|F^m_j|} >
\frac{1}{\mathtt{inf} \{ |F^m_i|, |F^m_k| \} } = 
d_g(F_i^m, F_j^m).
\]
The triangular inequality is thus satisfied.
\end{itemize}

If $m_i \not= m_j$, without loss of generality we assume $m_i < m_j$ in the following.
\begin{itemize}
\item If $m_k < m_i$, $m_i < m_k < m_j$, or $m_k > m_j$, obviously we have 
\[\displaystyle
d_g(F_i^\omega, F_k^\omega) = d_g(F_k^\omega, F_j^\omega) = 1.
\]    
Then, indeed:
\[\displaystyle
d_g(F_i^\omega, F_j^\omega) \leq 1 < 2 = 
d_g(F_i^\omega, F_k^\omega) + d_g(F_k^\omega, F_j^\omega). 
\]

\item If $m_k = m_i < m_j$, then we know  
\[\displaystyle
d_g(F_k^\omega, F_j^\omega) = 1.
\]
So, we have
\[\displaystyle
d_g(F_i^\omega, F_j^\omega) \leq 1 
< d_g(F_i^\omega, F_k^\omega) + 1
= d_g(F_i^\omega, F_k^\omega) + d_g(F_k^\omega, F_j^\omega). 
\]

\item If $m_k = m_j > m_i$, then we know  
\[\displaystyle
d_g(F_i^\omega, F_k^\omega) = 1.
\]
So, we have
\[\displaystyle
d_g(F_i^\omega, F_j^\omega) \leq 1 
< 1 + d_g(F_k^\omega, F_j^\omega)
= d_g(F_i^\omega, F_k^\omega) + d_g(F_k^\omega, F_j^\omega). 
\]
Therefore, again, the triangular inequality is satisfied.
\end{itemize}

\end{proof}

Lastly, we show that $d_g$ defined as such induces the topology over $F^\omega$.

\begin{proposition}
Topological space $\mathcal{T}_{d_g}^{F^\omega}$ on $F^\omega$ is induced by $d_g$.
\end{proposition}
\begin{proof}
Let $\displaystyle \epsilon = \frac{1}{1 + \mathtt{sup} \{ |F_i^{m_i}| : F_i^\omega \in F^\omega \}}$, $i \in \mathbb{Z}$, $0 \leq i < n$,
where $m_i$ is the smallest time index such that $|F_i^{m_i}| > 1$ in the infinite sequence $F_i^\omega$.
Then an open ball $B_{d_g}(u, \epsilon)$ is fine enough to isolate each element in $F^\epsilon$ because:
\[\displaystyle
\epsilon < \frac{1}{\mathtt{sup} \{ |F_i^{m_i}| : F_i^\omega \in F^\omega \}} 
\leq \frac{1}{\mathtt{inf} \{ |F_i^{m_i}| : F_i^\omega \in F^\omega \}}
\leq \frac{1}{\mathtt{inf} \{ |F_i^{m_i}|, |F_j^{m_j}| \}}
= d_g(F_i^\omega, F_j^\omega)
\]
for any $0 \leq i, j < n$.
Therefore, we found a basis $\mathcal{B}_g$ of space $\mathcal{T}_{d_g}^{F^\omega}$:
\[
\mathcal{B}_g = 
    \left\{   
        \left\{ F_0^\omega \right\},
        \dots,
        \left\{ F_{n-1}^\omega \right\}
    \right\}.
\]
\end{proof}

Now we have constructed three topological spaces for time-varied forks, for static forks, and for transactions among blockchains.
In the next section, we will extract the relationship across these three spaces.

\section{Commutative Morphism from Growing-fork Space to Transaction Space}
\label{sec:map}

We will show that a commutative morphism exists among the three spaces as the following.
That is, we will construct $g: \mathcal{T}^{F^\omega}_{d_g} \rightarrow \mathcal{T}^F_{d_f} $ and $h: \mathcal{T}^F_{d_f} \rightarrow \mathcal{T}^T_{d_t} $, respectively, 
and show that both maps are continuous.
We can then reason about the transaction status by studying the topology of growing-fork topology using the composite function $h \circ g$.

\begin{equation*}
\begin{tikzcd}
\mathcal{T}^{F^\omega}_{d_g}    \arrow [r, "g"]
            \arrow [dr, swap, "h \circ g"]
&
\mathcal{T}^F_{d_f}   
    \arrow [d, "h"]
\\
{}&
\mathcal{T}^T_{d_t}   
\end{tikzcd}
\end{equation*}

\subsection{Homeomorphism between Fork Space and Transaction Space}

Let $\mathcal{B}_f$ denote the basis for the fork space defined in Proposition~\ref{thm:fork_space},
and $\mathcal{B}_t$ denote the basis for the transaction space in Proposition~\ref{thm:txn_space}.
We use $\mathcal{F}$ to indicate the topological space induced from $\mathcal{B}_f$, i.e., $\mathcal{F} = \mathcal{T}_{d_f}^F$.
Similarly, we overuse $\mathcal{T}$ to indicate the topological space induced from $\mathcal{B}_t$, i.e., $\mathcal{T} = \mathcal{T}_{d_t}^C$, if it is clear from the context.
We will construct a bijective function $h: \mathcal{F} \rightarrow \mathcal{T}$,
and show that both $h$ and $h^{-1}$ are continuous.

Let $u \in \mathcal{F}$ be an open set.
This, by definition, means that $u$ is an arbitrary union of definite intersections among elements in $\mathcal{B}_f$.
Because $\mathcal{B}_f$ is discrete, $u$ is simply a set of arbitrary selections of $F$'s, denoted $\{F_{u_0}, \dots, F_{u_{m-1}}\}$, where $u_i < u_j$ if $i < j$, and $0 < u_m \leq n$.
Define $I$ to be the index set: $I = \{u_0, \dots, u_{m-1}\}$, and denote $u = F_I$.
Now, we define $v \in \mathcal{C}$ to be the set with the same set of indexed clusters: $v = C_I$.
Then, we construct the map as $h: u \mapsto v$.

\begin{proposition}
$h$ is a homeomorphism between $\mathcal{F}$ and $\mathcal{T}$.
\end{proposition}

\begin{proof}

We need to show that $h$ and $h^{-1}$ are (i) bijective and (ii) continuous.

\textbf{$h$ is bijective.}
If two elements $u, u' \in \mathcal{F}$ are distinct and map to the same $v \in \mathcal{T}$,
then the subsets $I_u$ and $I_{u'}$ in the index set $I$ consist of the same indices,
which means $u = u'$.
Therefore, $h$ is \textit{injective}.
Conversely, for any $v \in \mathcal{F}$, there must exist a subset $I_v$ by definition.
As a consequence, there must be an element in $\mathcal{F}$ that is induced by the subset $I_v$, again by definition.
Therefore, $h$ is \textit{surjective}.
As a result, $h$ is bijective.

\textbf{$h$ is continuous.}
$\forall v \in \mathcal{T}$, an open set in the transaction space, 
we will show that $h^{-1}(v)$ is an open set in the fork space $\mathcal{F}$.
Because $h$ is bijective, we know there must exist a unique value for $h^{-1}(v)$.
Let $u = h^{-1}(v)$, $u \in \mathcal{F}$.
Note that $\mathcal{F}$ is a discrete topology;
therefore $u$ is an open set in $\mathcal{F}$.

\textbf{$h^{-1}$ is continuous.}
The proof is similar to $h$, and we skip it here.

\end{proof}

Essentially, there is an equivalence between the transaction status and the transaction's proxies on involved clusters in terms of forks,
from a topological point of view.
Therefore, we just proved a stronger result than we need:
\[
\begin{tikzcd}
\mathcal{T}^F_{d_f}   
    \arrow [d, "h"]
\\
\mathcal{T}^T_{d_t}   
    \arrow [u, shift left, "h^{-1}"]
\end{tikzcd}
\]
Next, we will show an equivalence between such a ``static'' fork space and the entire ever-growing forks, topologically.

\subsection{From Growing-Fork Space to Fork Space}

This section will construct a continuous function $g: \mathcal{F}^\omega \rightarrow \mathcal{F}$ to ``flatten'' the infinite sequences of growing-fork spaces into a static fork space.
Intuitively, we will show that the time factor in the growing-fork space can be topologically preserved in the static fork space.

\begin{definition}
We define a map $g$ such that for a set of infinite sequences of forks $u \in F^\omega$,
$g(u)$ is the set of fork graphs of those transactions that incur the \textit{first} fork, 
denoted $v = g(u)$.
If the blockchain never spawns a fork, we define the corresponding element in $v$ as $\emptyset$.
\end{definition}

\begin{proposition}
$g$ is continuous.
\end{proposition}

\begin{proof}
Let $v \in \mathcal{F}$ be an open set, i.e., an arbitrary union of intersections of fork graphs.
Since $\mathcal{F}$ is a discrete topology,
$v$ is a subset of $F$, $v \in F$.
By definition of $g$, we know there exists a set of transactions $T_v$ who trigger the first forks on every fork element in $v$.

Recall that the topology $\mathcal{F}^\omega$ is induced by the open ball $B_{d_g}(u, \epsilon)$,
where $u \in \mathcal{F}$ and $\epsilon$ denotes the distance ``just'' smaller than the distance imposed by a transaction causing the first fork.
That is, every singleton element in the basis, i.e., $b \in \mathcal{B}_g$ is associated with a transaction $t_b$.
Therefore, we have
\[
t_b \in T_v \rightarrow b \in u.
\]
By definition, we know $b$ is an open set in $\mathcal{F}^\omega$.
From the above condition, $u$ is a union of $b$'s:
\[
u = \bigcup_{t_b \in T_v} b.
\]
Since a union of open sets is an open set,
$u$ is an open set in $\mathcal{F}^\omega$.

\end{proof}

Therefore, the following morphism holds:

\begin{equation*}
\begin{tikzcd}
\mathcal{T}^{F^\omega}_{d_g}    \arrow [r, "g"]
&
\mathcal{T}^F_{d_f}.   
\end{tikzcd}
\end{equation*}

Since both $g$ and $h$ are continuous, then the composite $h \circ g$ is also continuous and the following is true:
\begin{equation*}
\begin{tikzcd}
\mathcal{T}^{F^\omega}_{d_g}    \arrow [dr, "h \circ g"]
&
\\
{}&
\mathcal{T}^T_{d_t}. 
\end{tikzcd}
\end{equation*}

\section{Final Remark}
\label{sec:final}

This paper constructs a topological space upon the transactions among blockchains
and shows that the transaction space is homeomorphic to another topological space of static fork graphs induced by those transactions.
Further, this paper constructs a time-varying counterpart of the static fork graphs, namely the growing-fork topological space,
and shows that a continuous function exists from the growing-fork topology to the static forks.
Combined together, these results allow us to reason about the cross-blockchain transactions through the growing-fork topology, an intuitive representation of blockchains.
One of the most important applications based on these results, for example,
is the \textit{decidiability} of cross-blockchain transactions:
whether an arbitrary $n$-party cross-blockchain transaction is able to complete in finite time can now be reduced to a pure topological problem where we can call upon a vast literature of results in combinatorial and algebraic topology. 

\section*{Acknowledgment}

This work is supported by the U.S.
Department of Energy (DOE) under contract number DE-SC0020455.
This work is also supported by a research award from Amazon and a research award from Google.

\bibliographystyle{IEEEtran}
\bibliography{ref_new}

\end{document}